\documentclass[10pt,journal,letterpaper,compsoc]{IEEEtran}
\ifCLASSOPTIONcompsoc
\else
\fi
\ifCLASSINFOpdf
\else
\fi
\usepackage{graphicx}
\usepackage{algorithmic,algorithm}
\usepackage{amsmath}
\usepackage{amsfonts}
\usepackage{amssymb}
\usepackage{caption}
\usepackage{amsthm}
\usepackage{amsmath, amsfonts, amssymb}
\usepackage{algorithm, algorithmic}
\usepackage{graphicx, comment}
\usepackage{color,subfigure}

\def\0{{\mathbf{0}}}
\def\1{{\mathbf{1}}}

\def\lmat{\left(\begin{matrix}}
\def\rmat{\end{matrix}\right)}
\def\eqref#1{(\ref{#1})}

\newtheorem{theorem}{Theorem}

\newtheorem{lemma}{Lemma}

\newtheorem{corollary}{Corollary}

\usepackage{url}

\usepackage[algo2e,lined,boxed,commentsnumbered]{algorithm2e}

\DeclareMathOperator*{\argmin}{\arg\!\min}
\DeclareMathOperator*{\argmax}{\arg\!\max}

\begin{document}

%
\title{A Provably Efficient Online Collaborative Caching Algorithm for Multicell-Coordinated Systems}
%
%
%
%

\author{Ammar Gharaibeh,~\IEEEmembership{Student Member,~IEEE,} Abdallah Khreishah,~\IEEEmembership{Member,~IEEE,}
Bo Ji,~\IEEEmembership{Member,~IEEE}
and~Moussa Ayyash,~\IEEEmembership{Senior,~IEEE}
\IEEEcompsocitemizethanks{\IEEEcompsocthanksitem Ammar Gharaibeh and Abdallah Khreishah are with the Department of Electrical and Computer Engineering, New Jersey Institute of Technology, Newark, NJ 07102 USA, Bo Ji is with the Department of Computer and Information Sciences, Temple University, Philadelphia, PA, and Moussa Ayyash is with the Department of Information Studies, Chicago State University, IL.
\protect\\
E-mail: \{amg54,abdallah\}@njit.edu, boji@temple.edu, and mayyash@csu.edu.
}

\thanks{This research was supported in part by NSF grant ECCS 1331018.}
}

\IEEEcompsoctitleabstractindextext{
\begin{abstract}
Caching at the base stations brings the contents closer to the users, reduces the traffic through the backhaul links, and reduces the delay experienced by the cellular users. The cellular network operator may charge the content providers for caching their contents. Moreover, content providers may lose their users if the users are not getting their desired quality of service, such as maximum tolerable delay in Video on Demand services. In this paper, we study the collaborative caching problem for a multicell-coordinated system from the point of view of minimizing the total cost paid by the content providers. We formulate the problem as an Integer Linear Program and prove its NP-completeness. We also provide an online caching algorithm that does not require any knowledge about the contents popularities. We prove that the online algorithm achieves a competitive ratio of $\mathcal{O}(\log(n))$, and we show that the best competitive ratio that any online algorithm can achieve is $\Omega(\frac{\log(n)}{\log\log(n)})$. Therefore, our proposed caching algorithm is provably efficient. Through simulations, we show that our online algorithm performs very close to the optimal offline collaborative scheme, and can outperform it when contents popularities are not properly estimated.
\end{abstract}

}

\maketitle

\IEEEdisplaynotcompsoctitleabstractindextext

%
\IEEEpeerreviewmaketitle

\section{Introduction}\label{sec:Introduction}
Recently, content delivery has dominated the Internet traffic. Services like Video on Demand accounts for 54\% of the total Internet traffic, and the ratio is expected to grow to 71\% by the end of 2019 \cite{Cisco2014}. This expected increase motivated changes to the operations of cellular networks, as the current infrastructure cannot cope with this increase.

One way to handle the above challenge is to introduce caching at the base stations. Caching at the base stations can reduce the data traffic going through the backhaul links, reduce the time required for content delivery, and help in smoothing the traffic during peak hours. Thus, providing good caching techniques is of high importance. Note that the price of data storage devices is dramatically decreasing year by year.

Caching in general has been extensively studied. A few examples are those of \cite{Khreishah2015joint,Khreishah2014renewable,Gharaibeh2014asymptotically}. The works in \cite{ong2014fgpc,wang2013optimal,lee2013cache} study caching in Content Centric Networks \cite{jacobson2009networking} which has different settings than the multicell-coordinated system we consider in this paper. Caching in cellular networks has been studied under different settings and for different objectives. The objective of the work in \cite{blasco2014learning} is to minimize the delay through content popularity estimation in a single base station, while \cite{golrezaei2012femtocaching} uses caching helpers to achieve the same objective. The authors in \cite{pei2011cache,ahlehagh2012video,ostovari2013cache} consider hierarchical caching in cellular backhaul networks, while \cite{karamchandani2014hierarchical,maddah2013fundamental} take an information-theoretic approach to study hierarchical caching.

All of the above-mentioned studies do not take advantage of the backhaul links connecting the base stations, by not allowing the base stations to collaborate. When a base station receives a content request, it can retrieve a copy of the content from another base station that cached the content via the backhaul links, instead of retrieving it from the original Content Provider via the Internet. Therefore, collaboration among the base stations reduces the operational cost of the cellular network and enhances its performance. This will motivate the cellular network operator to deploy collaborative caching.

Collaborative caching has only been considered recently by proposing proactive caching schemes under offline settings. By proactive caching we mean that the caching decisions are made before the appearance of any request for any content, and by offline settings we mean that the caching scheme knows the exact popularities of the contents (i.e. the exact number of requests that will be made for every content). The fact that the exact contents popularities are hard to estimate accurately motivates us to propose a reactive caching algorithm under online settings. By reactive caching we mean that our algorithm makes a caching decision after the appearance of a request for a content, and by online settings we mean that our algorithm does not know the exact popularities of the contents. Nevertheless, our algorithm works on a per-request basis and caching decisions made by the algorithm does not require the knowledge of the contents popularities.

Some examples of proactive collaborative caching schemes are the works of \cite{wang2014cache,khreishah2015collaborative,wang2015framework,pantisano2014network}. The authors in \cite{wang2014cache} propose a proactive offline caching scheme, and develop a heuristic algorithm to minimize the content access delay of all users. The authors assume that content popularity is the same across different base stations. The authors in \cite{khreishah2015collaborative} propose an offline caching scheme to maximize the reward gained by the cellular network operator when the cache of each base station is limited. The work in \cite{wang2015framework} studies collaborative caching with the objective of minimizing either the inter ISP traffic, the intra ISP traffic, or the overall user delays, using a proactive offline caching scheme and a heuristic algorithm.

The differences between our work and the works of \cite{wang2014cache,khreishah2015collaborative,wang2015framework} are three folds:
\begin{itemize}
\item The works in \cite{wang2014cache,khreishah2015collaborative,wang2015framework}  all require the information of contents popularities, which may be hard to estimate accurately, while our algorithm does not need this information.
\item Based on the popularity information, the works in \cite{wang2014cache,khreishah2015collaborative,wang2015framework} all solve a static optimization problem. In contrast, our scheme works on a per-request basis where the requests for contents are revealed one by one, and the online algorithm has to make a decision based on the number of requests seen so far by the online algorithm.
\item Although the work in [16] proposes an online algorithm based on CCN caching, the algorithm is a simple Least Recently Used (LRU) caching algorithm and the authors do not provide any theoretical proof of the algorithm's performance.
\end{itemize}

In \cite{pantisano2014network}, the authors study collaborative caching among small base stations deployed in a single macro-cell. The objective of the study is to minimize the cellular network operational cost, given the cache size at each small base station and the bandwidth of the backhaul links. The authors also propose a combined offline algorithm and Least Frequently Used (LFU) replacement policy for in-network cache management, and prove that the ratio of the performance of the offline algorithm to the optimal offline algorithm is within a factor of $\beta$, which is linear in terms of the product of the number of potential collaborators, the number of requests for each user, and the number of cached contents at each small base station. Our work is different in that we propose an online caching algorithm that achieves a better performance ratio when compared to the optimal offline algorithm.

We consider collaborative caching at the base stations from a different perspective. Our objective is to minimize the overall cost paid by the Content Providers (CPs). We assume that caching a content at a base station incurs two types of costs. The first type is the storage cost, where CPs have to pay to the Cellular Network Operator (CNO) in exchange for caching their contents. This is motivated by the increasing trend of using in-network cloudlets, services, and middleboxes, in which the storage and computations are performed at small clouds installed in the routers or the base stations of the network \cite{chen2013packetcloud,satyanarayanan2009case,fesehaye2012impact,sherry2012making}. Moreover, the caching costs paid by the CPs motivate the CNOs to perform caching by providing them with an extra source of income.

The second type of cost is what we call User Attrition (UA) cost \cite{Wiki}. This cost represents the expected cost of losing users that are switching to other CPs because the users are not getting their desired Quality of Service (QoS). This is caused by the fact that the requested content is cached far away from the users. For example, users who are experiencing high delays when streaming a video from one CP may switch to another CP, which causes losses for the former CP. These two types of costs yield a tradeoff on where the CPs choose to cache their contents in order to minimize the total cost. In this paper, we formulate the problem of caching in a multicell-coordinated system as an optimization problem that minimizes the overall cost paid by the CPs, while satisfying the users' demands.

In the formulated optimization problem, we assume the exact knowledge of the contents popularities. Based on this knowledge, a proactive offline algorithm for collaborative caching can achieve the optimal solution, similar to \cite{wang2014cache,khreishah2015collaborative,wang2015framework,pantisano2014network}. Since in real life scenarios this knowledge is unavailable, an online algorithm is needed for caching at the base stations. In the online algorithm, a decision for caching at a base station is made when a content is requested, and the caching decision cannot be changed in the future because the CP has already paid the caching cost. To measure the performance of the online algorithm, we use the concept of competitive ratio. The competitive ratio is the ratio of the performance of the online algorithm to the performance of the optimal offline algorithm. In this paper, we present an online algorithm with a competitive ratio of $\mathcal{O}(\log(n))$, where $n$ is the total number of requests in the cellular network. The competitive ratio we obtain is close to the lower bound of $\Omega(\frac{\log (n)}{\log\log (n)})$ which we prove in Section \ref{subsec:LowerBound}, since $\log\log (n)$ is small even when $n$ is large (i.e. $\log\log (10^6) \approx 4.3$).

Specifically, we make the following contributions:
\begin{itemize}
\item We formulate the problem of collaborative caching in multicell-coordinated systems as an Integer Linear Program (ILP), aiming to minimize the overall cost paid by the CPs and we prove that the problem is NP-complete.

\item We provide an online algorithm for collaborative caching in multicell-coordinated system that does not require the knowledge of the contents popularities. We prove that the competitive ratio of the online algorithm is $\mathcal{O}(\log(n))$.
\item We prove that the best competitive ratio that any online algorithm can achieve is lower bounded by $\Omega(\frac{\log (n)}{\log\log (n)})$. 
\item We compare the performance of the offline collaborative caching and the online collaborative caching through extensive simulations. Simulation results reveal that the online algorithm can tolerate inaccuracies in measuring the contents popularities, which is not provided by any of the prior work \cite{Khreishah2015joint,Khreishah2014renewable,Gharaibeh2014asymptotically,golrezaei2012femtocaching,ostovari2013cache,maddah2013fundamental,wang2014cache,khreishah2015collaborative,wang2015framework,pantisano2014network}.
\end{itemize}

The rest of the paper is organized as follows: In Section \ref{sec:settings} we specify our settings. In Section \ref{sec:Formulation} we present the ILP problem formulation and prove its NP-completeness. In Section \ref{sec:Online} we present the online algorithm followed by the proof of its $\mathcal{O}(\log(n))$ competitive ratio. We also prove that the competitive ratio of any online algorithm is lower bounded by $\Omega(\frac{\log (n)}{\log\log (n)})$. Section \ref{sec:Simulation} presents our simulation results. We finally conclude the paper in Section \ref{sec:Conclusion}.

\section{Settings} \label{sec:settings}
We consider a cellular network consisting of a set $\mathcal{K} = \{1,2,\ldots,k,\ldots,K\}$ of cache-capable base stations connected to each other via backhaul links. The backhaul links also serve as a connection to the Internet through the cellular system gateway. In the rest of the paper, we use the words base station and cache interchangeably. Let $K+1$ denote the index of the contents providers servers located in the Internet. Fig. \ref{fig:model} provides an example of our system model. We have $\mathcal{M} = \{1,2,\ldots,j,\ldots,M\}$ contents with sizes $\mathcal{S} = \{s_1, s_2, \ldots, s_j, \ldots, s_M\}$ that can be requested by the users connected to the base stations. Let $\gamma_{ij}$ denote the number of requests for the $j$-th content generated by users in the $i$-th base station.

Due to the dramatic decrease in data storage prices, we assume that the cache capacity of each base station is unlimited. However, there is a unit cost $f_{kj}$ associated with caching the $j$-th content at the $k$-th base station. Having a caching cost will limit the number of contents cached at a base station. Moreover, as explained in Section~\ref{sec:Introduction}, there is an increasing trend of using in-network cloudlets in which the base station itself becomes a small cloud. Also, due to the fast development and cost reduction of storage devices, the cache size can be very large with low cost.

Let $T_{ij}^{k}$ denote the UA cost associated when the $i$-th base station retrieves the $j$-th content from the $k$-th base station, and let $T_{ij}^{K+1}$ denote the UA cost associated when the $i$-th base station retrieves the $j$-th content from the Internet. If we associate a cost between any two directly connected base stations, then for any two base stations $i, k$, $T_{ij}^{k}$ can be computed using the minimum cost path between $i$ and $k$, and thus satisfies the triangle inequality (\emph{i.e.} $T_{ij}^k \leq T_{ij}^{k'} + T_{k'j}^k$).

\begin{figure*}
\centering
\includegraphics[scale = 0.5]{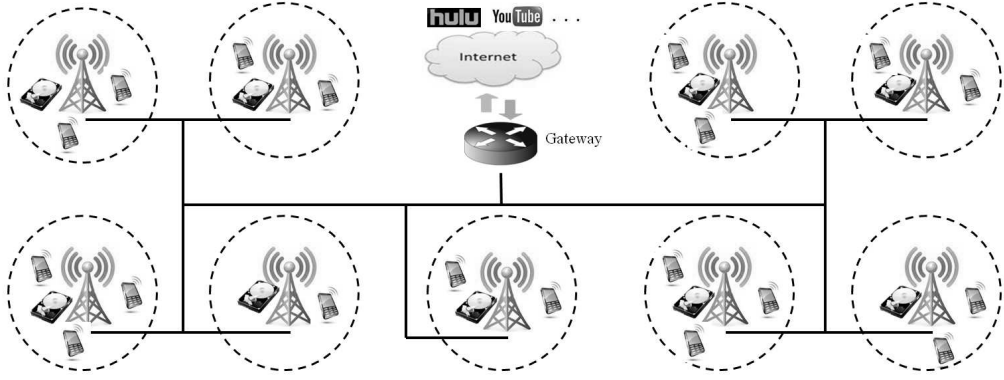}
\caption{Collaborative multicell-coordinated system.}
\label{fig:model}
\end{figure*}

Our objective is to find a caching setup that minimizes the aggregated caching and UA costs while satisfying the users' demands. We introduce the formulation of the optimization problem in the next section.

\section{Problem Formulation}\label{sec:Formulation}
In this section, we formulate the problem of collaborative multicell-coordinated system, followed by the proof of the problem's NP-completeness.

\subsection{Collaborative Case Formulation}\label{subsec:Collaborative}

Before presenting our formulation, we introduce the following variables:

$Y_{kj}  =\left\{ \begin{array}{rl}
1 &\mbox{if the $j$-th content is cached}\\
  &\mbox{at the $k$-th base station.}\\
0 &\mbox{otherwise.}
\end{array}\right.$

$X_{ij}^{k} = \left\{ \begin{array}{rl}
1 &\mbox{if the $i$-th base station retrieves the }\\
  &\mbox{$j$-th content from the $k$-th base station.} \\
0 &\mbox{otherwise.}
\end{array}\right.$

$X_{ij}^{K+1} = \left\{ \begin{array}{rl}
1 &\mbox{if the $i$-th base station retrieves the }\\
  &\mbox{$j$-th content from the Internet.} \\
0 &\mbox{otherwise.}
\end{array}\right.$

We formulate the problem as the following Integer Linear Program (ILP):

\begin{displaymath} \min \sum_{i=1}^{K}\sum_{k = 1}^{K+1}\sum_{j = 1}^{M} T_{ij}^{k}X_{ij}^{k}\gamma_{ij}s_j + \sum_{i = 1}^{K}\sum_{j = 1}^M f_{kj}Y_{kj}s_j \label{eqn:basic_formulation}
\end{displaymath}

Subject to

\begin{align}
& X_{ij}^{k} \leq Y_{kj}, \quad \forall i, j, k \label{const1}\\
& \sum_{k = 1}^{K+1} X_{ij}^{k} \geq 1_{\{\gamma_{ij} > 0\}} \quad \forall i, j \label{const2}
\end{align}

In the objective function of the above problem, the first term is the total UA cost, and the second term is the total caching cost. The first set of constraints ensures that a content can be retrieved from a base station only if the content is in the cache of that base station. The second set of constraints ensures that if a content is requested, then the content is served either from the cache of the local base station, the cache of a neighboring base station, or from the Internet.

In the following section, we present the proof of the problem's NP-completeness.

\subsection{NP-Completeness Proof}\label{subsec:NPcompleteProof}
In this section, we show that the ILP optimization problem presented in the previous section is NP-complete by proving the following theorem:

\begin{theorem}\label{theorem:theorem1}
The ILP optimization problem is NP-complete.
\end{theorem}

\begin{proof}
Since we have $M(K^2 + K)$ constraints, we can easily check the feasibility of any given solution in polynomial time by checking that the set of constraints \eqref{const1}-\eqref{const2} are not violated, thus the problem is in NP.

To prove that the problem is NP-hard, we reduce the set cover problem, which is known to be an NP-complete problem, to an instance of our problem. The set cover problem is defined as follows: Given a set of elements $U = \{u_1, u_2, \ldots, u_N\}$ called the universe, and a family of subsets of the elements in the universe $B = \{b_1, b_2, \ldots, b_L\}$, where each subset $b_k$ has a cost $a_k$. The objective is to find a collection of subsets of $B$, whose union is the universe, and its total cost is minimized.

The reduction from the set cover problem to our problem is done as follows: (1) The number of contents in our problem is set to 1. (2) Each element $u_i$ in the set cover problem is mapped to a base station requesting the content in our problem. (3) The caching cost of the $k$-th base station $f_k$ is set to $a_k$. (4) The demand $\gamma_i$ is set to 1 for all base stations. (5) The UA cost $T_i^k$ is set to 0 if $u_i \in b_k$ and is set to $2a_k$ otherwise.

Note that due to the reduction from the set cover problem to our problem, all the elements in the set cover problem are covered iff the total UA cost of the solution to our problem is 0. Now we prove that there exists a solution to the set cover problem of cost no greater than $\mathcal{A}$ iff there exists a solution to our problem with a cost no greater than $\mathcal{A}$.

The first direction is easy to see. If there exists a solution to the set cover problem with cost $\mathcal{A}$, then the sets form the caches and the total cost of our problem is $\mathcal{A}$. To prove the other direction, suppose we have a solution to our problem with a cost of $\mathcal{A} = \mathcal{A}^{Caching} + UA^{total}$, where $\mathcal{A}^{Caching}$ is the total caching cost and $UA^{total}$ is the total UA cost. Then we have the following two cases:
\begin{itemize}
\item The total UA cost is equal to 0. In this case, the selected caches yield the collection of sets for the set cover problem with cost $\mathcal{A}$.
\item The total UA cost is greater than 0. This means that there are some elements in the corresponding set cover problem that are still not covered. In this case, for each base station $i$ whose incurring a non-zero UA cost, we cache the content at a new cache $k$ (i.e. select a new subset $b_k$) such that $T_i^k = 0$. Since $f_k + T_i^k = a_k + 0 \leq 2a_k$, the new total cost $\mathcal{A}^{\prime} < \mathcal{A}$. Then go back to case 1.
\end{itemize}
\end{proof}

\section{Online Algorithm}\label{sec:Online}
\subsection{Definitions}
In the online version of the problem, the users' requests for contents are revealed one by one. The online algorithm has to make a decision to satisfy the request either by caching the content in a nearby base station or by retrieving the content from the Internet or from a base station that already has the content in its cache. The algorithm's decisions cannot be changed in the future, and the decisions must be made before the next request is revealed, so the online algorithm works without the knowledge of contents popularities as opposed to the offline ILP formulation.

To compare the performance (i.e. the total cost) of the online algorithm to that of the optimal offline algorithm, we use the concept of competitive ratio. Other works have used the concept of competitive ratio, but for different problems such as online routing \cite{jaillet2008generalized} or energy efficiency \cite{albers2007energy}. We define the competitive ratio as the worst-case ratio of the performance achieved by the online algorithm to the performance achieved by the optimal offline algorithm, \emph{i.e.}, if we denote the performance of the online algorithm by $\mathcal{P}_{on}$, and the performance of the offline algorithm by $\mathcal{P}_{off}$, then the competitive ratio is:
\begin{displaymath}
\sup_{t} \sup_{\substack{all\ input\\
sequences\ in\ [0,t]}} \frac{\mathcal{P}_{on}}{\mathcal{P}_{off}}.
\end{displaymath}
As the ratio gets closer to 1, the online performance gets closer to the offline performance. In other words, the smaller the competitive ratio, the better the online algorithm's performance.

Before we present the online algorithm, we point out the following observation. Due to the assumption that the cache capacities of the base stations are unlimited, the decision of caching a content at a base station is independent from the other contents, so we can view our problem as $M$ independent caching subproblems. Even with this decomposition, the proof of NP-completeness presented in Section \ref{subsec:NPcompleteProof} still holds for every subproblem, since we reduce the set cover problem to an instance of our problem that has one content. In the sequel, we only consider the $j$-th content, and hence, the content-index $j$ is omitted in the subscript of the notations throughout the rest of the paper. Moreover, in the online algorithm and in the proof of the competitive ratio, every term will be multiplied by the content size $s_j$, and hence, we set $s_j = 1$.

\subsection{The Online Algorithm}
The online algorithm for collaborative caching is presented in Algorithm \ref{alg:alg1}. We note that Algorithm \ref{alg:alg1} is for a specific content. For multiple contents, each content will have its own instance of the algorithm.

Throughout the algorithm, the following notations are used:
\begin{itemize}
\item $W$: the set of base stations already caching the content.
\item $V$: the set of requests processed so far by the algorithm.
\item $p(k)$: the potential function of the $k$-th base station.
\item For a cache $u$ and a request $v$ arriving at the $i$-th base station, $d(u,v) = T_{i}^{u}$
\item For a set of caches $\mathcal{X}$ and a request $v$ arriving at the $i$-th base station, $d(\mathcal{X},v) \equiv \min_{k \in \mathcal{X}}d(k,v)$.

\item $[x]^{+} \equiv \max\{x, 0\}$.
\end{itemize}

\begin{algorithm}
\caption{Online Collaborative Caching}
\label{alg:alg1}
\begin{algorithmic}[1]
\STATE{$W \leftarrow \{K+1\}, V \leftarrow \phi, Cost = 0$, initializePotentials()}
\FOR{each new request $v$ arriving at the $i$-th base station}
\IF{The $i$-th base station is already caching the content}
\STATE{Satisfy the request}
\ELSE
\STATE{$V \leftarrow V \cup \{v\}$}
\STATE{updatePotentials($W, v$)}
\STATE{$w \leftarrow \argmax_{k} (p(k) - f_{k})$}
\IF{$p(w) - f_{w} > 0$}
\STATE{$W \leftarrow W \cup \{w\}$}
\STATE{$Cost = Cost + f_w$}
\STATE{computeNewPotentials($W,V$)}
\ENDIF
\STATE{assign $v$ to $\alpha = \argmin_{k \in W}d(k,v)$}
\STATE{$Cost = Cost + T^{i}_{\alpha}$}
\ENDIF
\ENDFOR
\end{algorithmic}
\end{algorithm}

The functions initializePotentials(), updatePotentials(), and computeNewPotentials() used in Algorithm \ref{alg:alg1} are presented in Algorithm \ref{alg:alg2}.

\begin{algorithm}
\caption{Functions used in Algorithm \ref{alg:alg1}}
\label{alg:alg2}
\begin{algorithmic}
\STATE{initializePotentials()}
\FOR{all $k \in \mathcal{K} \cup \{K+1\}$}
\STATE{$p(k) = 0$}
\ENDFOR
\STATE{}
\STATE{updatePotentials($W, v$)}
\FOR{all $k \in \mathcal{K} \cup \{K+1\}$}
\STATE{$p(k) = p(k) + [d(W,v) - d(k,v)]^{+}$}
\ENDFOR
\STATE{}
\STATE{computeNewPotentials($W,V$)}
\FOR{all $k \in \mathcal{K} \cup \{K+1\}$}
\STATE{$p(k) = \sum_{v \in V}[d(W,v) - d(k,v)]^{+}$}
\ENDFOR
\end{algorithmic}
\end{algorithm}

The intuition behind Algorithm \ref{alg:alg1} is the following: we define a potential function for each base station. When a new request for a content arrives at a base station, the algorithm updates the potential function of each base station, which represents the total UA cost of the requests seen so far when that base station retrieves the content from the base station with the lowest UA cost (probably itself). The algorithm decides to cache the content at a base station when the potential of that base station exceeds its caching cost.

\subsection{Implementation and Complexity}
The execution of Algorithm \ref{alg:alg1} is done by the Mobility Management Entity (MME) of the cellular network, which acts as a centralized controller \cite{3gpp.36.300}. Different architectures for cellular networks like Software-Defined Cellular Networks \cite{li2012toward} take advantage of the centralized controller. The MME has access to the topology of the cellular network as well as the contents cached at each base station. The content providers pay the operators for running the caching algorithm at the MMEs and storing the contents at the base stations.

When a request for content arrives at a base station, the base station responds with the requested content if it already has a copy of the content in its cache. Otherwise, the base station sends a message indicating the requested content to the MME to execute the algorithm. Based on the available information to the MME (i.e. the topology of the cellular network as well as the contents cached at each base station), the MME runs the algorithm and decides whether the content has to be cached at a new base station $w$ or not. The MME then relays to the new base station $w$ (if any) the decision to cache the requested content and relays to the requesting base station the decision of which base station to retrieve the content from.

Next, we analyze the complexity of Algorithm \ref{alg:alg1}. Recall that $K$ denotes the number of base stations. The \emph{initializePotentials()} subroutine is executed once and has a complexity of $\mathcal{O}(K)$. For every new request, executing the \emph{updatePotentials()} subroutine, finding the base station with the maximum difference between the value of the potential function and the caching cost (line 5 in Algorithm \ref{alg:alg1}), and executing the \emph{computeNewPotentials()} subroutine each has a complexity of $\mathcal{O}(K)$. Note that the \emph{computeNewPotentials()} subroutine is executed only when the content is cached at a new cache, and hence is executed at most $K$ times. Therefore, for $n$ total number of requests, the overall complexity of implementing the algorithm is $\mathcal{O}(Kn + K^2)$.

In order to execute the algorithm, the MME needs to maintain two tables. The first table includes the value of the potential function of each base station. For the second table, it suffices to store the number of requests that arrived at each base station, since the value of $[d(W,v) - d(k,v)]^+$ in the \emph{computeNewPotentials()} subroutine is the same for requests arriving at the same base station. Therefore, the storage requirement needed at the MME is of order $\mathcal{O}(K)$. Note that the size of the tables is independent of the content's size. Therefore, as the size of the content becomes larger, which is the case for future content delivery traffic, the size of the tables will not grow.

\subsection{Preliminaries}
To compute the competitive ratio of the algorithm, we compare the algorithm's cost with the cost of the offline optimal solution. In the optimal offline solution, let $W^*$ denote the set of base stations caching the content. Then, the cost of the offline optimal solution is given by:
\begin{equation}
\mathcal{C}^{*} = \sum_{w \in W^*} f_{w} + \sum_{v \in V} d(W^*,v)
\end{equation}

Let the optimal solution $W^*$ consist of $l$ caches $c_1, c_2, \ldots, c_l$. In the optimal offline solution, each request is satisfied by retrieving the content from a cache. Hence, $W^*$ divides the requests into optimal clusters $C_1, C_2, \ldots, C_l$. For example, if the optimal solution decides to cache a content on three caches $c_1, c_4$, and $c_5$, then the first cluster $C_1$ consists of $c_1$ and all the base stations retrieving the content from $c_1$, the second cluster $C_2$ consists of $c_4$ and all the base stations retrieving the content from $c_4$, and so on.

\subsection{Proof Outline}
We start by proving that the algorithm maintains the invariant $p(k) \leq f_{k}$ for all $k \in \mathcal{K} \cup \{K+1\}$ (Lemma \ref{Lemma:lemma1}). Based on Lemma \ref{Lemma:lemma1} and the triangle inequality, we show that after $j$ requests from cluster $C_i$, there is a cache where the UA cost from the optimal cache $c_i$ is $\frac{1}{j}[f_{c_{i}} + 2\sum_{v \in C_i}d(W^*,v)]$ (Corollary 1). This is used to show that the UA cost of the algorithm is within a logarithmic factor of the total optimal cost $\mathcal{C}^{*}$ (Lemma \ref{Lemma:lemma2}).

For each new request $v$, we define a credit $\hat{c}(v)$. We show that $\hat{c}(v) = \min \{d(W,v), \min_{k} \{f_k - p(k) + d(k,v)\}\}$ (Lemma \ref{Lemma:lemma3}). We then show that the algorithm's total caching cost never exceeds the total credit of the requests in $V$ (Lemma \ref{Lemma:lemma4}). Using Corollary \ref{Corollary:corollary1}, we show that the total credit is within a logarithmic factor of the total optimal cost $\mathcal{C}^*$ (Lemma \ref{Lemma:lemma5}).

\subsection{The Proof}
\begin{lemma}\label{Lemma:lemma1}
$p(k) \leq f_k, \forall k \in \mathcal{K} \cup \{K+1\}$.
\end{lemma}

\begin{proof}
We prove this lemma by induction on the number of requests considered by the online algorithm. For the first request, the invariant holds since $p(k) = 0$ for all $k \in \mathcal{K} \cup \{K+1\}$. We inductively assume that the invariant holds just before a new request $v$ arrives and prove that the invariant holds after $v$ is assigned to retrieve the content from a cache.

Let $W$ be the set of caches that have cached the content, and let $V$ be the set of requests considered so far by the algorithm. Let $p(k) = \sum_{v' \in V}[d(W,v') - d(k,v')]^{+} \leq f_k$ be the potential of cache $k$ just before the new request $v$ arrives. Let $p'(k)$ be the potential after the subroutine \emph{updatePotentials()} in the algorithm is executed. Finally, let $p''(k)$ be the potential of cache $k$ after the request $v$ is assigned to a cache. We want to prove that $p''(k) \leq f_k$ for all $k$.

We have two cases:
\begin{itemize}
\item If the new request $v$ does not change $W$ (i.e. the new request did not cause the content to be cached at a new cache), then from the algorithm:
\begin{displaymath}
p'(k) - f_k \leq 0 \qquad \forall k
\end{displaymath}
and $p''(k) = p'(k)$. Therefore, $p''(k) \leq f_k$ for all $k$.
\item If the new request $v$ causes the content to be cached at a new cache $w$, then
\begin{align}
0 < p'(w) - f_w &= [d(W,v) - d(w,v)]^{+} + p(w) - f_w \nonumber\\
 & \leq [d(W,v) - d(w,v)]^{+} \nonumber
\end{align}
where the first inequality holds because the content is cached at $w$, the next equality holds from the definition of $p'(w)$, the last inequality holds from the induction hypothesis $p(w) \leq f_w$. Therefore, $[d(W,v) - d(w,v)]^{+} \geq p'(w) - f_w > 0$. This implies that $d(W,v) > d(w,v)$, which means that $v$ will be assigned to $w$ and
\begin{equation}
d(W \cup \{w\},v) = d(w,v) \label{eq:proof1}
\end{equation}

From here, we have two cases:
\begin{itemize}
\item For all caches $k$ where $d(k,v) < d(w,v)$, we have $[d(W\cup \{w\},v) - d(k,v)]^{+} = d(w,v) - d(k,v)$.

Using \eqref{eq:proof1}, we get that
\begin{align}
d(W,v) &- d(w,v) \nonumber\\
&\geq p'(w) - f_w \nonumber\\
 &\geq p'(k) - f_k \nonumber\\
 &= [d(W,v) - d(k,v)]^{+} + p(k) - f_k \nonumber\\
 &\geq d(W,v) - d(k,v) + p(k) - f_k \nonumber
\end{align}
where the second inequality follows from the fact that the content is cached at $w$, and the first equality follows from the definition of $p'(k)$. Therefore, $d(W,v) - d(w,v) \geq d(W,v) - d(k,v) + p(k) - f_k$. Rearranging the terms we get
\begin{align}
0 \geq & d(w,v) - d(k,v) + p(k) - f_k \nonumber\\
 \geq & d(w,v) - d(k,v)\nonumber\\
&+ \sum_{v' \in V}[d(W,v') - d(k,v')]^{+} - f_k \nonumber\\
 \geq & d(w,v) - d(k,v)\nonumber\\
&+ \sum_{v' \in V}[d(W\cup \{w\},v') - d(k,v')]^{+} - f_k \nonumber\\
  \geq & p''(k) - f_k \nonumber
\end{align}
where the last inequality follows from the definition of $p''(k)$.

\item For all caches $k$ where $d(k,v) \geq d(w,v)$, we have $[d(W\cup \{w\},v) - d(k,v)]^{+} = 0$. Therefore,
\begin{align}
p''(k) &= \sum_{v' \in V \cup \{v\}}[d(W\cup \{w\},v') - d(k,v')]^{+} \nonumber\\
 & = \sum_{v' \in V}[d(W\cup \{w\},v') - d(k,v')]^{+} \nonumber\\
 & \leq \sum_{v' \in V}[d(W,v') - d(k,v')]^{+} \nonumber\\
 & = p(k) \leq f_k \nonumber
\end{align}
where the first equality holds from the definition of $p''(k)$, the second equality follows since $[d(W\cup \{w\},v) - d(k,v)]^{+} = 0$, the third inequality follows since $d(W\cup \{w\},v') \leq d(W,v'), \forall v'$, and the last equality follows from the definition of $p(k)$ before the request $v$ appears.

\end{itemize}
\end{itemize}
\end{proof}

\begin{corollary}\label{Corollary:corollary1}
Let $V$ be the request set, and $W$ the set of caches caching the content after all requests in $V$ have been considered. Then for every optimal cluster $C_i$ with cache $c_i$,

\begin{equation}
|V \cap C_i|d(W,c_i) \leq f_{c_i} + 2\sum_{v \in C_i} d(W^*,v) \nonumber
\end{equation}
\end{corollary}

\begin{proof}
for cache $c_i$ we have:
\begin{align}
p(c_i) &= \sum_{v \in V}[d(W,v) - d(c_i,v)]^{+} \nonumber \\
&\geq \sum_{v \in V \cap C_i}[d(W,v) - d(c_i,v)]\nonumber \\
&\geq \sum_{v \in V \cap C_i}[d(W,c_i) - d(c_i,v) - d(c_i,v)]\nonumber \\
&\geq \sum_{v \in V \cap C_i}d(W,c_i) - 2\sum_{v \in V \cap C_i}d(c_i,v)\nonumber \\
&\geq \sum_{v \in V \cap C_i}d(W,c_i) - 2\sum_{v \in C_i}d(c_i,v)\nonumber
\end{align}
where the second inequality is obtained by using the triangle inequality (i.e. $d(W,v) \geq d(W,c_i) - d(c_i,v)$).

Using the invariant $f_{c_i} \geq p(c_i)$ from Lemma 1 and rearranging the terms, we get:
\begin{equation}
|V \cap C_i|d(W,c_i) \leq f_{c_i} + 2\sum_{v \in C_i} d(W^*,v) \nonumber
\end{equation}
where $d(W^*,v) = d(c_i,v)$ for all $v \in C_i$ follows from the definition of cluster $C_i$.
\end{proof}

In the next lemma, we use Corollary \ref{Corollary:corollary1} to bound the UA cost incurred by the online algorithm.
\begin{lemma}\label{Lemma:lemma2}
Let $\sum_{v \in V}d(W,v)$ denote the total UA cost incurred by the online algorithm. Then
\begin{align}
\sum_{v \in V}d(W,v) &\leq \log(n+1)\sum_{w \in W^*}f_w \nonumber \\
&+ (2\log(n+1) + 1)\sum_{v \in V}d(W^*,v) \nonumber
\end{align}
\end{lemma}

\begin{proof}
Let $C_i$ be an optimal cluster with cache $c_i$. Let $n_i \equiv |C_i|$ be the number of requests in $C_i$, and let $v_1,v_2,\ldots,v_{n_i}$ be the requests in $C_i$ in the order considered by the algorithm.

For each request $v_j$, let $W_{v_j}$ be the set of caches caching the content at $v_j$'s assignment time. Then, using triangle inequality we have
\begin{displaymath}
d(W_{v_j},v_j) \leq d(W_{v_j},c_i) + d(c_i,v_j)
\end{displaymath}

From Corollary \ref{Corollary:corollary1}, we have
\begin{displaymath}
d(W_{v_j},c_i) \leq \frac{1}{j}[f_{c_i} + 2\sum_{v \in C_i}d(W^*,v)]
\end{displaymath}

Therefore
\begin{displaymath}
d(W_{v_j},v_j) \leq \frac{1}{j}[f_{c_i} + 2\sum_{v \in C_i}d(W^*,v)] + d(c_i,v_j)
\end{displaymath}

Summing over all $v_j \in C_i$ we get that $\sum_{j = 1}^{n_i}d(W_{v_j},v_j)$
\begin{align}
\leq &[f_{c_i} + 2\sum_{v \in C_i}d(W^*,v)]\sum_{j = 1}^{n_i}\frac{1}{j} + \sum_{j = 1}^{n_i}d(c_i,v_j) \nonumber \\
\leq &\log(n_i + 1)f_{c_i} \nonumber\\
&+ (2\log(n_i + 1) + 1)\sum_{j = 1}^{n_i}d(c_i,v_j)\nonumber
\end{align}
The lemma follows by summing over all clusters.
\end{proof}

The next 3 lemmas are used to bound the caching cost incurred by the online algorithm. To do this, we define a credit $\hat{c}(v)$ to each new request $v$. We show that $\hat{c}(v) = \min\{d(W,v), \min_{k}\{f_k - p(k) +d(k,v)\}\}$. Then we show that the total caching cost is upper bounded by the total credits of all requests, which in turn is within a logarithmic factor of the optimal offline cost $\mathcal{C}^*$.

\begin{lemma}\label{Lemma:lemma3}
For each new request $v$, $\hat{c}(v) = f_w - p(w) +d(w,v)$ if $v$ causes the content to be cached at $w$, and $\hat{c}(v) = d(W,v)$ otherwise.
\end{lemma}

\begin{proof}
Let $p(k)$ denote the potential function of the cache $k$ just before a new request $v$ arrives, and let $p'(k) = p(k) + d(W,v) - d(k,v)$ be the potential function after the subroutine \emph{updatePotential()} is executed.

Now, if the new request $v$ causes the content to be cached at $w$, then from the algorithm,
\begin{align}
0 < p'(w) - f_w &= p(w) + [d(W,v) - d(w,v)]^+ - f_w \nonumber\\
&= p(w) + d(W,v) - d(w,v) - f_w \nonumber\\
&\geq p(k) + [d(W,v) - d(k,v)]^+ - f_k \nonumber\\
&\geq p(k) + d(W,v) - d(k,v) - f_k \nonumber
\end{align}
The first inequality holds because the content is cached at $w$. The first equality follows from the definition of $p'(w)$. The second equality follows by using (4)(see the proof of Lemma 1, page 6). The second inequality holds because $p'(w) - f_w \geq p'(k) - f_k, \forall k$ and from the definition of $p'(k)$.
Therefore, $p(w) + d(W,v) - d(w,v) - f_w \geq p(k) + d(W,v) - d(k,v) - f_k$. Rearranging the terms we get that $f_w - p(w) + d(w,v) \leq f_k - p(k) + d(k,v)$. We also have $p(w) + d(W,v) - d(w,v) - f_w > 0$. Therefore, $\hat{c}(v) = f_w - p(w) +d(w,v)$ if $v$ causes the content to be cached at $w$.

On the other hand, if the new request $v$ does not cause an additional copy of the content to be cached, then from the algorithm,
\begin{align}
&p'(k) - f_k \leq 0, \forall k \nonumber \\
&p(k) + d(W,v) - d(k,v) - f_k \leq 0, \forall k \nonumber \\
&d(W,v) \leq f_k - p(k) + d(k,v), \forall k \nonumber
\end{align}
Therefore, $\hat{c}(v) = d(W,v)$.
\end{proof}

In the next lemma, we show that the total caching cost incurred by the online algorithm is upper bounded by the total credit of the requests in $V$.
\begin{lemma}\label{Lemma:lemma4}
Let $V$ be the set of requests, and let $W$ be the set of caches caching the content after all requests in $V$ have been considered. Then
\begin{displaymath}
\sum_{w \in W}f_w \leq \sum_{v \in V}\hat{c}(v)
\end{displaymath}
\end{lemma}

\begin{proof}
We prove this lemma by a potential function argument. We define the potential function $\Phi = \sum_{v \in V}d(W,v)$ and calculate the change $\Delta \Phi$ in the value of the potential function when a new request $v$ is considered. Let $p(k)$ be the potential of each cache $k$ just before the new request $v$ arrives.

If the new request $v$ does not cause the content to be cached at a new cache (i.e $W$ is not changed), then $\Delta \Phi = d(W,v) = \hat{c}(v)$ by Lemma \ref{Lemma:lemma3}. Otherwise, if $v$ causes the content to be cached at $w$, then $d(W\cup \{w\},v) = d(w,v)$ (recall \eqref{eq:proof1} in Lemma \ref{Lemma:lemma1}) for all $v \in V$, and $d(W,v) - d(W\cup \{w\},v) = [d(W,v) - d(w,v)]^{+}$. Therefore,
\begin{align}
\Delta \Phi &= d(w,v) - \sum_{v' \in V}[d(W,v') - d(W\cup \{w\},v')] \nonumber\\
 & = d(w,v) - \sum_{v' \in V}[d(W,v') - d(w,v')]^{+} \nonumber\\
 & = d(w,v) - p(w) \nonumber
\end{align}
From Lemma \ref{Lemma:lemma3}, we have $\hat{c}(v) = f_w - p(w) + d(w,v) = f_w + \Delta \Phi$.
Therefore, $\sum_{v \in V}\hat{c}(v) = \Phi + \sum_{w \in W}f_w$. The lemma follows since $\Phi \geq 0$.
\end{proof}

In the next lemma, we use Corollary \ref{Corollary:corollary1} to upper bound the total credit of the requests in $V$.
\begin{lemma}\label{Lemma:lemma5}
Let $V$ be the set of requests. Then
\begin{displaymath}
\sum_{v \in V}\hat{c}(v) \leq (\log(n) + 1)\sum_{w \in W^*}f_w + (2\log(n) + 1)\sum_{v \in V}d(W^*,v)
\end{displaymath}
\end{lemma}
\begin{proof}
Let $C_i$ be an optimal cluster with cache $c_i$ and $n_i \equiv |C_i|$ be the number of requests in $C_i$. For each request $v_j \in C_i$, let $W_{v_j}$ be the set of caches caching the content just before $v_j$ arrives. Note that $d(W^*,v_j) = d(c_i,v_j), \forall v_j \in C_i$.

The credit of each request $v_j$ is $\hat{c}(v_j) \leq \min\{d(W_{v_j},v_j), f_{c_i} + d(W^*,v_j)\}$ (using Lemma \ref{Lemma:lemma3}). For the first request, $\hat{c}(v_j) \leq f_{c_i} + d(W^*,v_j)$. For the remaining requests $v_j, j \geq 2$, we use Corollary \ref{Corollary:corollary1} to get
\begin{align}
\hat{c}(v_j) &\leq d(W_{v_j},v_j) \nonumber \\
&\leq d(W_{v_j},c_i) + d(c_i,v_j) \nonumber \\
&\leq \frac{1}{j-1}[f_{c_i} + 2\sum_{v_j \in C_i}d(W^*,v_j)] + d(c_i,v_j) \nonumber
\end{align}

Summing over all $v_j$, we get
\begin{align}
\sum_{j = 1}^{n_i}\hat{c}(v_j) &\leq f_{c_i} + [f_{c_i} + 2\sum_{v_j \in C_i}d(W^*,v_j)]\sum_{j = 2}^{n_i}\frac{1}{j - 1} \nonumber \\
&+ \sum_{j = 1}^{n_i}d(W^*,v_j) \nonumber \\
&\leq (\log(n_i) + 1)f_{c_i} + (2\log(n_i) + 1)\sum_{j = 1}^{n_i}d(W^*,v_j) \nonumber
\end{align}

The lemma follows by summing over all clusters.
\end{proof}

Now we are ready to prove the algorithm's competitive ratio
\begin{theorem}\label{theorem:theorem2}
The competitive ratio of the online algorithm is no more than $4\log(n+1) + 2$.
\end{theorem}

\begin{proof}
From Lemma \ref{Lemma:lemma2}, we have
\begin{align}
\sum_{v \in V}d(W,v) &\leq \log(n+1)\sum_{w \in W^*}f_w \nonumber \\
&+ (2\log(n+1) + 1)\sum_{v \in V}d(W^*,v) \nonumber
\end{align}
and from Lemmas \ref{Lemma:lemma4} and \ref{Lemma:lemma5}, we have
\begin{displaymath}
\sum_{w \in W}f_w \leq (\log(n) + 1)\sum_{w \in W^*}f_w + (2\log (n) + 1)\sum_{v \in V}d(W^*,v)
\end{displaymath}

Combining the two bounds we get
\begin{align}
\sum_{w \in W}f_w + \sum_{v \in V}d(W,v) &\leq (2\log(n+1) + 1)\sum_{w \in W^*}f_w \nonumber \\
&+ (4\log(n+1) + 2)\sum_{v \in V}d(W^*,v) \nonumber \\
&\leq (4\log(n+1) + 2)\mathcal{C}^* \nonumber
\end{align}
\end{proof}

\subsection{Lower Bound}\label{subsec:LowerBound}
To prove the lower bound of the competitive ratio of any online algorithm under our settings, we measure the competitive ratio of the online algorithm against an oblivious adversary. For a deterministic algorithm, the adversary knows how the algorithm works, so the adversary can always generate an input sequence such that the deterministic algorithm performs worst on that input. For a randomized algorithm, the adversary knows the algorithm's code, but does not know the randomized result of the randomized algorithm, so the performance of the randomized algorithm is not worse than the performance of the deterministic algorithm against the same adversary. This means that a lower bound on the competitive ratio of the randomized algorithm is also a lower bound on the competitive ratio of the deterministic algorithm. In the next theorem, we show that the competitive ratio of any randomized online algorithm is lower bounded by $\Omega(\frac{\log(n)}{\log(\log(n))})$.

\begin{figure*}
\centering
\includegraphics[scale = 0.8]{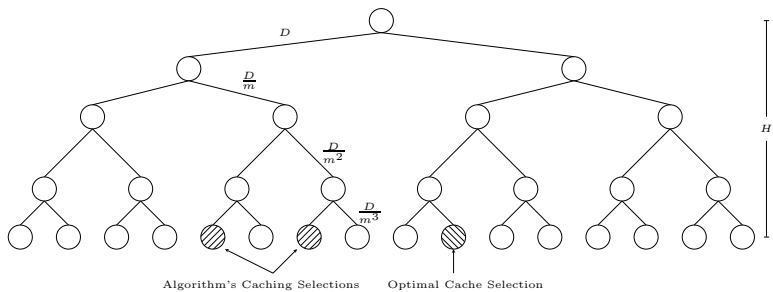}
\caption{Tree structure used in the proof of Theorem \ref{theorem:theorem3}.}
\label{fig:Tree}
\end{figure*}

The proof of the lower bound is done by showing an example, such that any online algorithm executed on this example cannot get a competitive ratio less than $\frac{\log(n)}{\log(\log(n))}$. Therefore, the best competitive ratio achieved by any online algorithm is $\frac{\log(n)}{\log(\log(n))}$. Moreover, the performance we consider is the worst-case performance, so one counter-example is sufficient.

\begin{theorem}\label{theorem:theorem3}
Under our settings, the best competitive ratio achieved by any randomized online algorithm against an oblivious adversary is lower bounded by $\Omega(\frac{\log(n)}{\log(\log(n))})$.
\end{theorem}

\begin{proof}
We prove the theorem through an example. Let $\mathcal{T}$ be a complete binary tree of height $H$ such that:
\begin{itemize}
\item Each vertex represents a base station.
\item Each edge between two vertices represent the UA cost between the corresponding directly-connected base stations.
\item The UA cost from the root to each of its children is $D$.
\item On every path from the root to a leaf, the cost drops by a factor of $m$ on every step.
\item The UA cost from vertex $i$ to vertex $k$ is the aggregated cost of the edges from $i$ to $k$.
\item The caching cost of every non-leaf vertex is set to infinity, while the caching cost of the leaves is set to $f$.
\end{itemize}
The height of a vertex is the number of edges on the path to the root. The cost from a vertex of height $h$ to each of its children is $\frac{D}{m^{h}}$. The tree is shown in Fig. \ref{fig:Tree}.

For a vertex $z$, let $\mathcal{T}_{z}$ denote the subtree rooted at vertex $z$. We observe the following:
\begin{itemize}
\item The cost from a vertex of height $h$ to any vertex in $\mathcal{T}_{z}$ is at most $\frac{mD}{(m-1)m^{h}}$, which is the UA cost from vertex $z$ to a leaf in $\mathcal{T}_{z}$.
\item The cost from a vertex $z$ of height $h$ to any vertex not in $\mathcal{T}_{z}$ is at least $\frac{D}{m^{h-1}}$, which is the cost from a vertex $z$ to its parent.
\end{itemize}

By Yao's principle \cite{bartal1996probabilistic}, it suffices to show that there is a probability distribution over the request sequence, for which the ratio of the expected cost of any deterministic online algorithm to the expected optimal cost is $\Omega(\frac{\log(n)}{\log(\log(n))})$.

To define an appropriate probability distribution, we divide the request sequence into $H+1$ phases. Phase 0 consists of 1 request located at the root. After the end of phase $h, 0 \leq h \leq H$, if $z_{h}$ is not a leaf, the adversary selects $z_{h+1}$ uniformly at random and independently between the two children of $z_{h}$. Phase $h+1$ consists of $m^{h+1}$ requests located at $z_{h+1}$.

The total number of requests is at most $\frac{m}{m-1}m^H$, which must not exceed $n$. The optimal solution is to cache the content at $z_H$, and each phase $h$ except for the last one, incurs a UA cost of at most $\frac{mD}{m-1}$. Therefore, the optimal total cost is at most $f + H\frac{mD}{m-1}$.

Now, let $Alg$ be any deterministic online algorithm, and let $h, 0 \leq h \leq H-1$, be any phase except for the last one. We fix the adversary's random choices $z_0, \ldots, z_{h}$ up to the end of phase $h$, and consider the expected cost paid by $Alg$ for requests and caches not in $\mathcal{T}_{z_{h+1}}$.

If $Alg$ did not cache the content at any cache located in $\mathcal{T}_{z_{h}}$ at the moment the first requests in $z_{h+1}$ arrives, then the content was cached at a cache located in $\mathcal{T} / \mathcal{T}_{z_{h}}$ from a previous phase. Therefore, the UA cost for the requests located at $z_{h} \in \mathcal{T}_{z_{h}} / \mathcal{T}_{z_{h+1}}$ is at least $\frac{m^{h}D}{m^{h-1}} = mD$, since these requests has to retrieve the content by going through the parent $z_h$. Otherwise, since $z_{h+1}$ is selected uniformly at random and independently between $z_{h}$'s children, then, with a probability of at least $1/2$, there is at least one cache located in $\mathcal{T}_{z_{h}} / \mathcal{T}_{z_{h+1}}$ that cached the content. Therefore for every fixed choice of $z_0, \ldots, z_{h}$, the expected cost paid by $Alg$ for requests and caches not in $\mathcal{T}_{z_{h+1}}$ is at least $\min\{mD, f/2\}$, in addition to the costs for requests and caches not located in $\mathcal{T}_{z_{h}}$.

Hence, at the beginning of phase $h, 0 \leq h \leq H$, the expected cost paid by $Alg$ for requests and caches not in $\mathcal{T}_{z_{h}}$ is at least $h\min\{mD, f/2\}$. For the last phase, $Alg$ incurs an additional cost of at least $\min\{mD, f\}$ for requests and caches in $\mathcal{T}_{z_{H}}$.

For $m = H$ and $f = HD$, the total expected cost of $Alg$ is at least $h\min\{HD, HD/2\} + min\{HD, HD\} \leq h\frac{HD}{2} + HD = HD\frac{h+2}{2}$, while the optimal cost is at most $HD\frac{2H - 1}{H - 1}$. Hence the competitive ratio is lower bounded by $\Omega(H)$. We also have the constraint that the total number of requests, which is at most $\frac{H^{H+1}}{H-1}$ must not exceed $n$. Setting $H = \lfloor{\frac{\log (n)}{\log\log(n)}}\rfloor$ yields the claimed lower bound (see the Appendix).
\end{proof}

\section{Simulation Results}\label{sec:Simulation}
\subsection{Settings}
In this section, we compare three caching schemes: the online collaborative scheme described in Section \ref{sec:Online}, the optimal offline collaborative caching scheme represented by the ILP formulation described in Section \ref{sec:Formulation} and computed using CPLEX \cite{cplex}, and the optimal offline non-collaborative caching scheme, which does not allow collaboration among the base stations, and can be formulated in a similar way to our ILP formulation.

\begin{figure*}
\centering
\includegraphics[scale = 0.45]{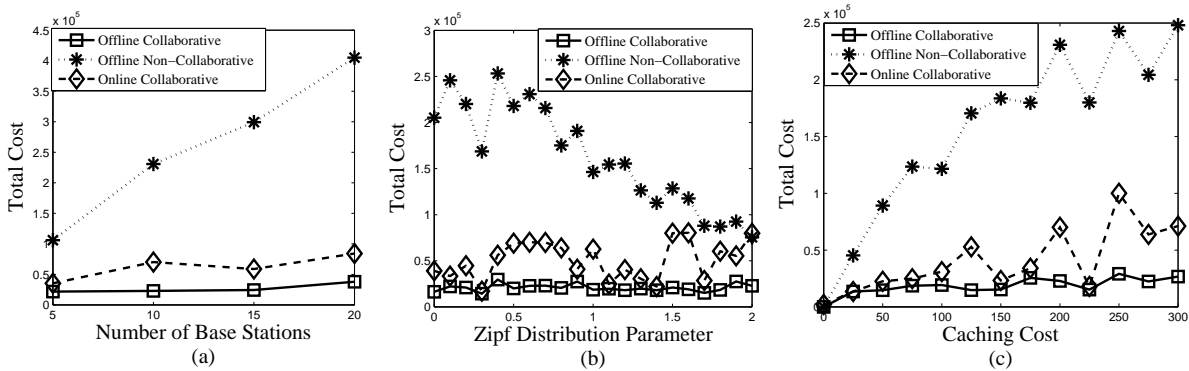}
\caption{Total cost of all schemes.}
\label{fig:Fig1}
\end{figure*}

The simulations are run on a random topology, where the base stations are uniformly distributed in a square area of size $50\times 50$ square kilometers and there is a link between two base stations if the distance between them is less than a certain threshold. If we associate a UA cost between any two directly connected base stations, then for any two base stations $i, k$, $T_{ij}^{k}$ can be computed using the path with the minimum UA cost between $i$ and $k$. We set the number of contents $M$ to 20, where each content has a size chosen uniformly at random from the set \{10, 11, \ldots, 20\}MB. The popularity of each content at each base station is chosen according to a Zipf distribution \cite{breslau1999web}, with parameter $\zeta$, where the popularity of a content of rank $j$ is given as $\frac{1/j^\zeta}{\sum_{m=1}^M 1/m^\zeta}$. We assume that content ranking in a base station is different and independent from the ranking in other base stations (\emph{i.e.} content $j$ may be ranked first in a base station but ranked fifth in another base station). The results in all of the figures are the average of 100 runs.

\subsection{Results with Accurate Estimation of Content Popularities}
We study the effect of changing different parameters on the cost of all schemes. In Fig. \ref{fig:Fig1}(a), we study the effect of changing the number of base stations. As can be observed from the figure, the cost of the non-collaborative scheme is at least 4 times and 3 times that of the cost of the offline and the online collaborative schemes, respectively. We also observe that the offline non-collaborative scheme cost increases at a higher rate than the collaborative schemes when we increase the number of base stations. The reason behind the above observations is that a base station in the non-collaborative scheme has to retrieve the content from the Internet, if the base station did not cache the content. On the other hand, the base station in the collaborative scheme can retrieve the content from a nearby base station that has the requested content. This shows the scalability of the collaborative caching schemes. Therefore, it is crucial to enable collaboration among the base stations, which is not done in most of the previous work as described above.

In Fig. \ref{fig:Fig1}(b), we study the effect of changing the Zipf distribution parameter $\zeta$. First, we observe that the cost of the non-collaborative scheme is at least 3 times the cost of the offline collaborative scheme, and varies between 1 to 3 times the cost of the online collaborative scheme, as the non-collaborative scheme cannot retrieve a content from a nearby base station. Second, we observe that the offline non-collaborative scheme cost decreases as $\zeta$ increases. This is because as $\zeta$ increases, less number of contents are requested more often, and the non-collaborative scheme caches the most popular files. Third, we observe that the offline collaborative scheme cost does not change much as $\zeta$ increases. This is because the offline collaborative scheme can retrieve contents from nearby base stations.

In Fig. \ref{fig:Fig1}(c), we study the effect of changing the average caching cost. First, we observe that the cost of the non-collaborative scheme is at least 2 times the cost of both of the offline and the online collaborative schemes. Second, we observe that the total cost of the non-collaborative scheme increases at a higher rate than the collaborative schemes as the average caching cost increases. The reason behind both observations is that the non-collaborative scheme does not allow content retrieval from nearby base stations, which means that each content may be cached at more than one base station, which increases the total caching cost. On the other hand, the collaborative schemes tend to cache each content at a single base station, and all other base stations can retrieve the content without caching an additional copy of the content or retrieve it from the Internet.

Moreover, in Fig. \ref{fig:Fig1}(c), we start with a caching cost that is less than the attrition cost (i.e. $f_{kj} < T_{ij}^{k}$), then we increase the caching cost until it becomes larger than the attrition cost (i.e. $f_{kj} > T_{ij}^{k}$). We can see a tradeoff between the caching cost and the attrition cost in this figure for the non-collaborative caching scheme. In the non-collaborative scheme, the content is either cached or retrieved from the Internet. When the caching cost is low, all base stations tend to cache the content and thus achieving low cost. As the caching cost increases, the base stations tend to retrieve the contents from the Internet instead of caching. However, retrieving from the Internet is costly. On the other hand, when the caching cost is large, the collaborative caching scheme tends to retrieve from nearby base stations. Thus the collaborative caching scheme achieves lower cost than the non-collaborative caching scheme.

From all of the plots in Fig. \ref{fig:Fig1}, we note that enabling collaboration among base stations has a significant impact on the cost reduction. We also note that the cost of the online collaborative scheme is very close to the cost of the optimal offline collaborative scheme, with a maximum degradation of three folds, and that the online collaborative scheme can achieve a cost reduction of four folds over the cost of the offline non-collaborative scheme.

In Fig. \ref{fig:Fig2}, we study the impact of increasing the average number of users at each base station on the cost of each scheme. As can be seen from the figure, the cost of the online collaborative scheme is less than that of the cost of the non-collaborative caching scheme. This is due to the collaborative property of the online collaborative scheme, where a content can be retrieved from a nearby base station.

\begin{figure}
\centering
\includegraphics[scale = 0.4]{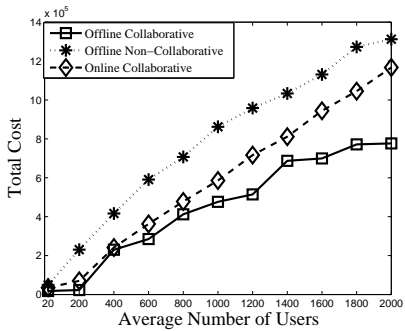}
\caption{Total cost of all schemes vs average number of users.}
\label{fig:Fig2}
\end{figure}

\begin{figure}
\centering
\includegraphics[scale = 0.4]{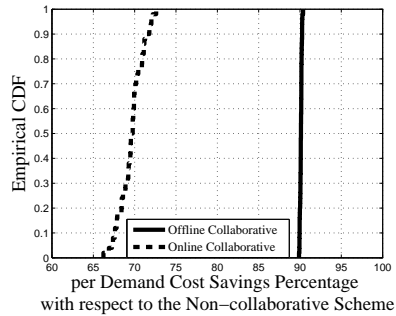}
\caption{The empirical CDF of the per demand cost savings percentage with respect to the non-collaborative scheme.}
\label{fig:cdf1}
\end{figure}

In Fig. \ref{fig:cdf1}, we measure the per demand cost savings percentage of the collaborative schemes with respect to the non-collaborative scheme. Here, we measure the cost of the collaborative schemes for 100 different sets of demands. For each set of demand, we normalize the cost of the collaborative schemes with respect to the cost of the non-collaborative scheme, and then subtract it from 1. Denote the cost of the offline collaborative and the offline non-collaborative for the $r$-th set of demands as $\mathcal{C}_{col}(r)$ and $\mathcal{C}_{non}(r)$, respectively. We compute the per demand cost savings as $R_{col}(r) = (1 - \frac{\mathcal{C}_{col}(r)}{\mathcal{C}_{non}(r)})\times 100\%$. After that, the empirical CDF of the vector $[R_{col}(1), R_{col}(2), \ldots, R_{col}(100)]$ for the 100 sets of demands is plotted. We do the same process for the online collaborative scheme. From the figure, we observe that the relative cost savings between the offline collaborative scheme and the online collaborative scheme is similar among all sets of demands. We also observe that the online collaborative cost savings varies between 65\% to 75\%. The reason is that the non-collaborative scheme cannot retrieve cached contents from another base station.

\subsection{Results with Errors in Estimating the Contents Popularities}

\begin{figure*}
\centering
\includegraphics[scale = 0.45]{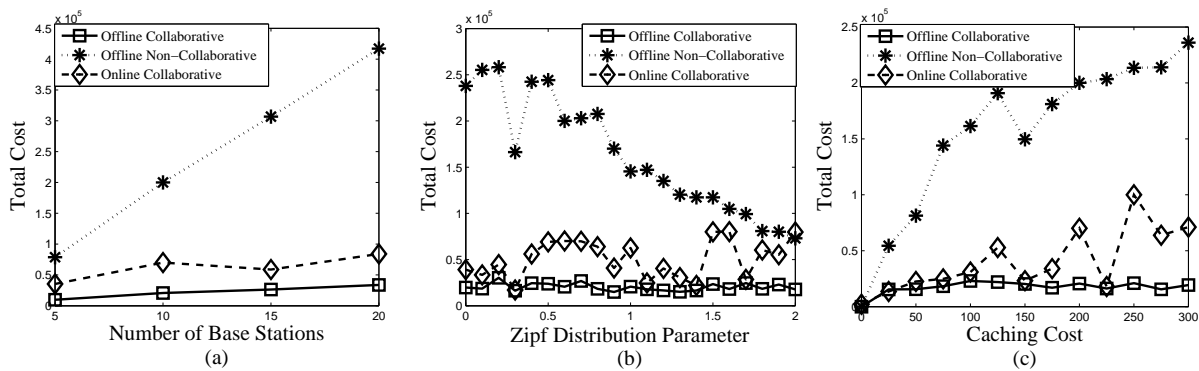}
\caption{Total cost of all schemes with 50\% error margin in popularity estimation.}
\label{fig:Fig3}
\end{figure*}

In Fig. \ref{fig:Fig3}, we repeat the simulations as in Fig. \ref{fig:Fig1}, when a 50\% error margin is introduced to the popularity estimation. Formally speaking, we generate two sets of requests, the estimated requests set and the actual requests set, where the estimated requests $\hat{\gamma}_{ij}$ are chosen randomly from a uniform distribution in the range $[0.5\gamma_{ij}, 1.5\gamma_{ij}]$. We use the set of estimated requests to solve the collaborative and the non-collaborative caching optimization problems, then we calculate the total cost based on the answer of the optimization problem and the set of actual requests. In the online collaborative scheme, the total cost is calculated using the actual requests, since the online collaborative scheme works when a new request for a content arrives at a base station.

In Fig. \ref{fig:Fig3}(a), we study the effect of changing the number of base stations. As can be observed from the figure, the offline and online collaborative schemes can achieve a cost reduction of at least 500\% and 100\% over the cost of the non-collaborative scheme, respectively. We also observe that the offline non-collaborative scheme cost grows at a higher rate than the collaborative schemes as the number of base stations increases. The reason is that an error in estimating the contents popularities may cause the non-collaborative scheme not to cache the correct contents, in which case the contents are retrieved from the Internet. On the other hand, the collaborative scheme can retrieve the contents from a nearby base station that has the requested contents.

\begin{figure}
\centering
\includegraphics[scale = 0.4]{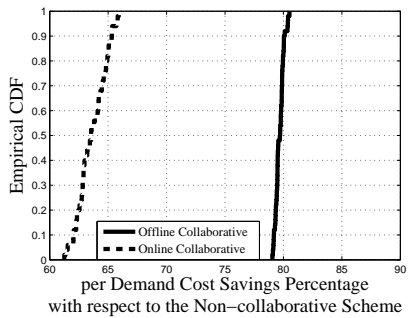}
\caption{The empirical CDF of the per demand cost savings percentage with respect to the non-collaborative scheme with 50\% error margin in popularity estimation.}
\label{fig:cdf2}
\end{figure}

\begin{figure}
\centering
\includegraphics[scale = 0.4]{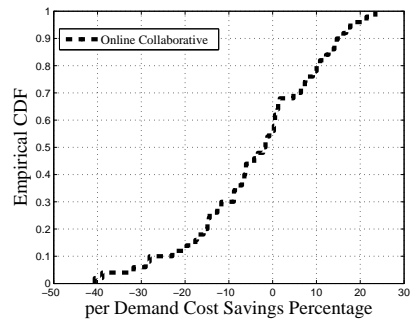}
\caption{The empirical CDF of the per demand cost savings percentage of the online collaborative scheme with respect to the offline collaborative scheme with 50\% error margin in popularity estimation.}
\label{fig:cdf3}
\end{figure}

In Fig. \ref{fig:Fig3}(b), we study the effect of changing the Zipf distribution parameter $\zeta$. First, we observe that the cost of the non-collaborative scheme is increased by at least 3-fold compared to the cost of the offline collaborative scheme, and the cost of the non-collaborative scheme can increase up to 10-fold compared to the cost of the online collaborative scheme. This is because the non-collaborative scheme cannot retrieve a content from a nearby base station. Second, we observe that the offline non-collaborative scheme cost decreases as $\zeta$ increases. This is because as $\zeta$ increases, less number of contents are requested more often, and the non-collaborative scheme caches the most popular files. Last, we observe that the offline collaborative scheme cost does not change much as $\zeta$ increases. This is because the offline collaborative scheme can retrieve contents from nearby base stations.

\begin{figure*}
\centering
\includegraphics[scale = 0.45]{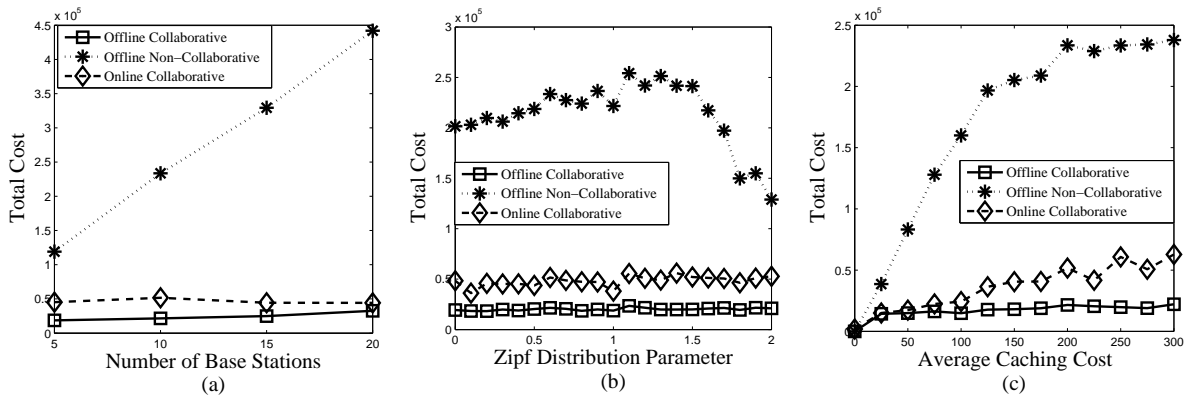}
\caption{Total cost of all schemes with errors in content ranking estimation.}
\label{fig:Fig4}
\end{figure*}

In Fig. \ref{fig:Fig3}(c), we study the effect of changing the average caching cost. First, we observe that the cost of the non-collaborative scheme is at least 3 times the cost of both of the offline and the online collaborative schemes. Second, we observe that the total cost of the non-collaborative scheme increases at a higher rate than the collaborative schemes as the caching cost increases. The reason behind both observations is that the non-collaborative scheme prohibits retrieving the contents from nearby base stations, which means that each content may be cached at more than one base station, which increases the total caching cost. On the other hand, the collaborative schemes tend to cache each content at a single base station, and all other base stations can retrieve the content without caching an additional copy of the content or retrieve it from the Internet.

In Fig. \ref{fig:cdf2}, we repeat the simulation as in Fig. \ref{fig:cdf1}, when a 50\% error margin is introduced to the popularity estimation. We observe that due to the error margin in popularity estimation, there is a small variance in the cost savings of the offline collaborative scheme over the offline non-collaborative scheme. We also observe that the online collaborative scheme cost savings varies between 71\% to 76\% over the cost of the non-collaborative scheme in all sets of demands.

Fig. \ref{fig:cdf3} shows the per demand cost savings percentage of the online collaborative schemes with respect to the offline collaborative scheme when the number of base stations is set to 10, the Zipf distribution parameter $\zeta$ is set to 1.1, the average caching cost is set to 200, and a 50\% error margin in popularity estimation is introduced. The per demand cost savings percentage is calculated in a similar manner to that in Fig. \ref{fig:cdf1}. From the figure, we note that the online collaborative scheme can outperform the offline collaborative scheme in 40\% of the demand sets, and it can achieve around 22\% of cost savings over the offline collaborative scheme. This is because the offline collaborative scheme may consider an unpopular content to be popular due to the inaccuracies in estimating the contents popularities. On the other hand, the online collaborative scheme is not affected by the inaccuracies in estimating content popularities, as it makes a caching decision when a new request for a content arrives at a base station.

In Fig. \ref{fig:Fig4}, we repeat the simulations as in Fig. \ref{fig:Fig1} when we introduce errors in estimating the contents ranking. The results are averaged over 20 runs. In Zipf popularity distribution, lower ranking number means higher popularity. To introduce the errors in estimating contents ranking, we generate two sets of requests, the estimated requests set and the actual requests set. The estimated requests set is generated by changing the contents ranking at each base station randomly. For example, if the actual ranking of a content at a base station is fifth, then the estimated ranking of the same content at the same base station may be first.

From the figure, we note that the total cost of the non-collaborative scheme varies between 2 times to 9 times of the total cost of the collaborative schemes as the number of base station increases (Fig. \ref{fig:Fig4}(a)). As we change the Zipf distribution parameter, the total cost of the non-collaborative scheme is at least twice the total cost of the collaborative schemes (Fig. \ref{fig:Fig4}(b)). Finally, as the average caching cost increases, the total cost of the non-collaborative scheme varies between 100\% to 400\% of the total cost of the collaborative schemes (Fig. \ref{fig:Fig4}(c)). The reason behind all these observation is that the collaborative schemes have a lower total UA cost since the collaborative schemes can retrieve the content from a nearby base station.

Fig. \ref{fig:cdf4} is similar to Fig. \ref{fig:cdf2} when we introduce errors in estimating the contents ranking. The results are for 20 different sets of demands. From the figure, we note that the relative cost savings of the offline collaborative scheme varies between 87\% to 92\% over all sets of demands, while the relative cost savings of the online collaborative scheme varies between 62\% to 92\% over all sets of demands.

In Fig. \ref{fig:cdf5}, we measure the relative cost savings of the online collaborative scheme to the offline collaborative scheme over 20 different sets of demands, when we introduce errors in estimating the contents ranking. From the figure, we note that in 30\% of the demand sets, the total cost of the online collaborative scheme is less than the total cost of the offline collaborative scheme. This is mainly due to the inaccuracies in estimating the contents ranking.

\section{Conclusion}\label{sec:Conclusion}
In this paper, we study the problem of content caching in a collaborative multicell-coordinated system, with the objective of minimizing the total costs paid by the content providers. We formulate the problem of collaborative caching as an optimization problem, and we prove that it is NP-complete. We also provide an online algorithm for the problem. The online algorithm does not require any knowledge about the content popularities. Through extensive simulations, we show that the collaborative caching schemes provide higher savings than the non-collaborative caching scheme, which means that applying the simple online algorithm is better than solving the non-collaborative optimization problem. The simulations also show that our online caching scheme can also outperform the optimal offline collaborative scheme when there are inaccuracies in estimating the contents popularities.

\begin{figure}
\centering
\includegraphics[scale = 0.4]{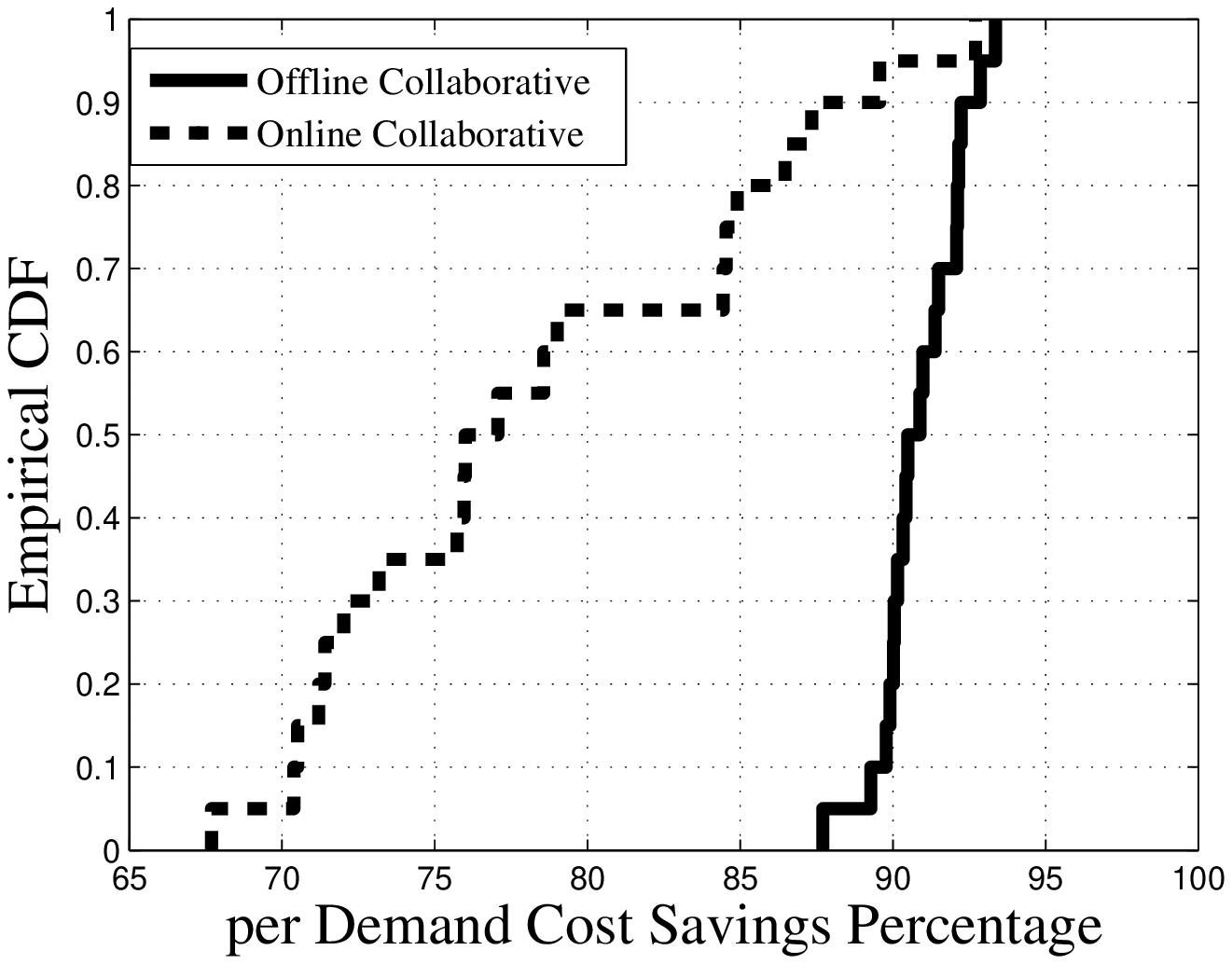}
\caption{The empirical CDF of the per demand cost savings percentage of the online collaborative scheme with respect to the offline collaborative scheme with errors in content ranking estimation.}
\label{fig:cdf4}
\end{figure}

\begin{figure}
\centering
\includegraphics[scale = 0.4]{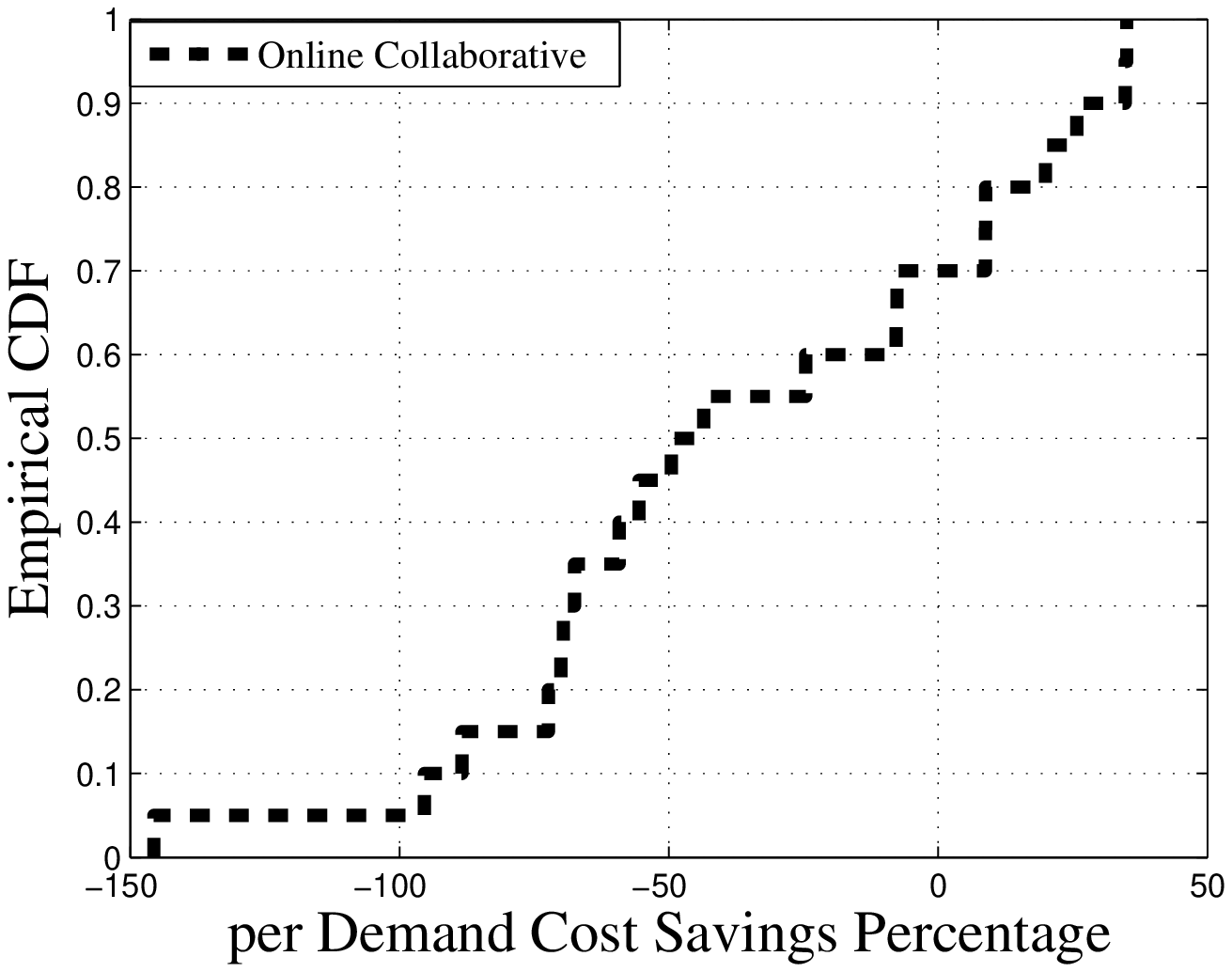}
\caption{The empirical CDF of the per demand cost savings percentage of the online collaborative scheme with respect to the offline collaborative scheme with errors in content ranking estimation.}
\label{fig:cdf5}
\end{figure}

\section*{Appendix}
To show that setting $H = \frac{\log(n)}{\log\log(n)}$ will satisfy the inequality $\frac{H^{H+1}}{H-1} \leq n$, we show that setting $H = \frac{\log(n)}{\log\log(n)}$ will satisfy $\frac{H^{H+1}}{H/2} \leq n$.
\begin{align}
\leq \frac{H^{H+1}}{H/2} &\leq n \nonumber\\
2H^H &\leq n \nonumber\\
H + H\log(H) &\leq \log(n) \nonumber
\end{align}
Setting $H = \frac{\log(n)}{\log\log(n)}$ yields
\begin{align}
\frac{\log(n)}{\log\log(n)} + \frac{\log(n)}{\log\log(n)} \log(\frac{\log(n)}{\log\log(n)}) &\leq \log(n) \nonumber\\
1 + \log(\frac{\log(n)}{\log\log(n)}) & \leq \log\log(n) \nonumber\\
1 + \log\log(n) - \log\log\log(n) &\leq \log\log(n) \nonumber
\end{align}
The last inequality holds when $n \geq 16$.

\bibliography{biblo}
\bibliographystyle{IEEEtran}




\end{document}